\documentclass[11pt]{article}
\usepackage[margin=1in, centering]{geometry}
\usepackage{mathtools}
\usepackage{amssymb}
\usepackage{amsthm}
\usepackage{xcolor}
\usepackage{bbm}
\usepackage[parfill]{parskip}

\usepackage[colorlinks=true, citecolor=blue]{hyperref}
\usepackage[capitalise]{cleveref}
\usepackage{enumerate}
\usepackage{cite}

\newtheorem{theorem}{Theorem}[section]
\newtheorem{claim}[theorem]{Claim}
\newtheorem{lemma}[theorem]{Lemma}
\newtheorem{definition}[theorem]{Definition}
\newtheorem{corollary}[theorem]{Corollary}

\newcommand{\F}{\mathbb{F}}
\newcommand{\N}{\mathbb{N}}
\newcommand{\R}{\mathbb{R}}
\DeclareMathOperator*{\E}{\mathbb{E}}

\newcommand{\1}{ \mathbbm{1}}

\newcommand{\eps}{\varepsilon}
\newcommand{\set}[2]{\left\{#1 \: : \: #2\right\}}
\newcommand{\ip}[2]{\left\langle#1 ,#2 \right\rangle}
\newcommand{\norm}[1]{\left\|#1 \right\|}
\newcommand{\paren}[1]{\left(#1 \right)}
\newcommand{\ind}[1]{\1 \! \! \paren{#1}}

\newcommand{\bits}{\{0,1\}}
\newcommand{\CI}{\text{CI}}
\newcommand{\IP}{\text{IP}}
\newcommand{\Hom}{\text{Hom}}
\newcommand{\Mon}{\text{Mon}}
\newcommand{\utk}[1]{U_{2,k}\left( #1 \right)}

\newcommand{\cH}{\mathcal{H}}
\newcommand{\col}{\text{Col}}
\newcommand{\cols}{\text{Cols}}
\newcommand{\CIC}{\overline{CI}}

\newcommand{\exactly}{\mathsf{ExactlyN}}

\newcommand{\ignore}[1]{{}}

\begin{document}

\title{Explicit separations between randomized and deterministic Number-on-Forehead communication}

\author{
Zander Kelley\thanks{
Department of Computer Science, University of Illinois at Urbana-Champaign. Supported by NSF CAREER award 2047310. Email: \texttt{awk2@illinois.edu}
}
\and
Shachar Lovett\thanks{
Department of Computer Science and Engineering, University of California, San Diego. Supported by NSF DMS award 1953928, NSF CCF award 2006443, and a Simons investigator award. Email: \texttt{slovett@ucsd.edu}
}
\and
Raghu Meka\thanks{Department of Computer Science, University of California, Los Angeles. Supported by NSF AF 2007682 and NSF Collaborative Research Award 2217033. Email: \texttt{raghum@cs.ucla.edu}
}}

\maketitle

\begin{abstract}
We study the power of randomness in the Number-on-Forehead (NOF) model in communication complexity. We construct an explicit 3-player function $f:[N]^3 \to \{0,1\}$, such that: (i) there exist a randomized NOF protocol computing it that sends a constant number of bits; but (ii) any deterministic or nondeterministic NOF protocol computing it requires sending about $(\log N)^{1/3}$ many bits. This exponentially improves upon the previously best-known such separation. At the core of our proof is an extension of a recent result of the first and third authors on sets of integers without 3-term arithmetic progressions into a non-arithmetic setting.
\end{abstract}

\section{Introduction}

Number-on-Forehead (NOF) communication introduced by Chandra, Furst, and Lipton \cite{chandra1983multi} is a central model in communication complexity. In its basic form, there are three parties, Alice, Bob, and Charlie, who each have an input $x,y,z$ respectively written on their forehead, and their goal is to communicate and compute a function $f(x,y,z)$. The twist, of course, is that Alice only knows the inputs $y,z$, Bob the inputs $x,z$, and Charlie the inputs $x,y$. The main goal is to understand how much communication is needed to compute the function $f$. 

Since its introduction, NOF communication complexity has been extensively studied and is known to have connections to circuit lower bounds, data structure lower bounds, and additive combinatorics \cite{chandra1983multi,raz2000bns,beame2006strong,beame2007lower,david2009improved,lee2009disjointness,beame2010separating,beame2012multiparty,lee2016hellinger,linial2018communication,shraibman2018note,alon2020number,linial2021larger,linial2021improved,alman2023matrix}, among others. In this work, we study the relative power of randomized vs. deterministic and non-deterministic protocols in the NOF model. 

Like standard two-party communication, given a function $f: X \times Y \times Z \rightarrow \{0,1\}$, one can study NOF communication complexity under several models: deterministic, non-deterministic and randomized protocols. We formally define these later, and note here that they are the natural analogs of the two-party definitions, where for randomized protocols we assume access to public randomness. 

The (complement of) set-disjointness function provides an explicit example of a function which is easy for nondeterministic protocols but is hard for deterministic or randomized protocols, see \cite{chattopadhyay2010story} for a survey on the history of this problem. Our focus in this paper is on the other direction; namely explicit functions which are easy for randomized protocols, but are hard for deterministic and non-deterministic protocols.

Despite the rich literature on NOF protocols, the relative power of randomness for NOF protocols is poorly understood. Beame {et al.}  \cite{beame2010separating} showed that there are functions $f:[N]^3 \rightarrow \{0,1\}$ whose randomized communication complexity is $O(1)$ but whose deterministic (or even non-deterministic) communication complexity is as large as possible, namely $\Omega(\log N)$. However, this separation is existential (based on a counting argument), and no explicit function is known to strongly separate deterministic and randomized NOF protocols. This is in stark contrast to the two-party case, where the equality function\footnote{The equality function is $EQ:[N]^2 \rightarrow \{0,1\}$, $EQ(x,y) = \1[x=y]$.} has randomized complexity $O(1)$ and deterministic complexity $\Omega(\log N)$.

The best known explicit separation is given by the \emph{Exactly-N} problem, already considered in the seminal work \cite{chandra1983multi} that defined the NOF model: $\exactly:[N]^3 \rightarrow \{0,1\}$ is defined by $\exactly(x,y,z) = 1$ if and only if $x+y + z = N$. The randomized communication complexity of $\exactly$ is $O(1)$ while the best known lower bound on its deterministic communication complexity is $\Omega(\log \log N)$ \cite{beigel2006multiparty,linial2021larger}.
The interest in this function stems from the fact that deterministic lower bounds for it are equivalent to the famous corners problem in additive combinatorics, which asks for the largest subset of $[N]^2$ without a corner\footnote{A corner is a triple of points $(x,y),(x+h,y),(x,y+h)$.}.
The conjectured bounds in additive combinatorics would imply a lower bound of $\Omega(\sqrt{\log N})$ for the deterministic communication complexity of $\exactly$; but as mentioned, the known lower bounds are exponentially far from it.

Our main result is an explicit construction of a 3-party function $D:[N]^3 \to \bits$, with constant randomized NOF communication complexity, and with deterministic (and even non-deterministic) NOF communication complexity polynomial in the input length, concretely $\Omega((\log N)^{1/3})$. 

\begin{theorem}\label{main}
Let $q$ be a prime power and $k$ a large enough constant ($k=34$ suffices). Let $N=q^k$, and identify $[N]$ with $\F_q^k$. Consider the following 3-player function:
$$
D(x,y,z) = \1[\ip{x}{y}=\ip{x}{z}=\ip{y}{z}],
$$
where $\ip{\cdot}{\cdot}$ denotes the standard inner-product in $\F_q^k$. Then:
\begin{enumerate}
    \item The randomized NOF communication complexity of $D$ (with public randomness) is $O(1)$.
    \item The deterministic or non-deterministic NOF communication complexity of $D$ is $\Omega((\log N)^{1/3})$.
\end{enumerate}
\end{theorem}

One of the challenges in showing a result as above is the dearth of techniques for proving lower bounds in the NOF communication complexity. By and large, the main technique for showing lower bounds in the 3-party NOF model is the discrepancy method as introduced in \cite{babai1989multiparty}. However, these techniques typically also lower bound the randomized communication complexity. As such, they are not immediately useful for separating deterministic and randomized communication complexity. 

Inspired by recent progress on the \emph{three-term arithmetic progressions} problem \cite{kelley2023strong}, we introduce a new technique for lower bounding deterministic NOF communication complexity. We are optimistic that this technique will find potential applications (for instance for the $\exactly$ problem), beyond the application in this paper. Below we first give a high-level overview of the main ideas, and then delve into the technical details.

\section{Proof overview}
We next give a high-level description of the proof of \Cref{main} highlighting the new techniques we introduce. 
The proof involves two modular steps:
\begin{itemize}
    \item We introduce a notion of \emph{pseudorandomness against cubes} for sets $D \subset [N]^3$. We show that computing membership in any \emph{sparse} pseudorandom set $D$ (i.e., $|D| < N^{3-c}$ for a constant $c > 0$) is hard for deterministic and non-deterministic NOF protocols. 
    \item We then construct sparse pseudorandom sets as above, which are in addition easy for randomized protocols. 
\end{itemize}

The first step above is more intricate, and we describe it first. 

\subsection{Pseudorandomness against cubes}

A cube $C \subset [N]^3$ is a set of the form $C=X \times Y \times Z$ for some subsets $X, Y, Z \subseteq [N]$.
At a high level, we say $D \subset [N]^3$ is pseudorandom if for any large enough cube $C$, the intersection $D \cap C$ \emph{looks random} in two ways: 

\begin{itemize}
    \item Its  density when restricted to the cube is comparable to its overall density.
    \item The marginals of $D$ on the faces of $C$ are close to uniform. 
\end{itemize}

Below we use the following convention: we identify a set $D \subset [N]^3$ with its indicator function $D:[N]^3 \to \bits$. The density of a set is $\E[D]=|D|/N^3$.

\begin{definition}[Pseudorandom set with respect to large cubes]
Let $D \subset [N]^3$ of density $\mu=\E[D]$, and let $\gamma>0$. Given a cube $C=X \times Y \times Z \subset [N]^3$, we say that $D$ is $\gamma$-pseudorandom with respect to $C$ if the following conditions hold:
    \begin{enumerate}
        \item $\E_{x,y,z \in X\times Y\times Z} [D(x,y,z)] = \mu (1 \pm \gamma)$.
        \item $\E_{x,y,z,z' \in X \times Y \times Z \times Z} [D(x,y,z) D(x,y,z')] = \mu^2(1 \pm \gamma)$,
        as well as the analogous conditions involving $(x,x',y,z)$ and $(x,y,y',z)$.
    \end{enumerate}
We say that $D$ is $\gamma$-pseudorandom with respect to large cubes if $D$ is $\gamma$-pseudorandom with respect to any cube $C$ of size $|C| \ge \gamma N^3$.
\end{definition}

We show that the above notion of pseudorandomness along with \emph{sparsity} suffices for hardness against non-deterministic NOF protocols (which include as a special case also deterministic NOF protocols). 

\begin{theorem}
\label{nondet-lb-psd-sparse}
Let $D \subset [N]^3$ be a set of size $|D| \le N^{3-c}$ which is  $(N^{-c})$-pseudorandom with respect to large cubes, for some constant $c>0$.
Then any non-deterministic NOF protocol which computes $D$ must send $\Omega((\log N)^{1/3})$ communication.
\end{theorem}

The above result follows as a corollary to a more general result on how pseudorandom sets as defined above interact with \emph{cylinder intersections} which we next describe.

\subsection{From pseudorandomness against cubes to cylinder intersections}

We use the well-known connection between NOF protocols and \emph{cylinder intersections}. 

\begin{definition}[Cylinder intersection]
A set $T \subset [N]^3$ is a cylinder intersection, if there are sets $S_1,S_2,S_3 \subset [N]^2$ such that
$$
T = \{(x,y,z): (x,y) \in S_1, (x,z) \in S_2, (y,z) \in S_3\}.
$$
Equivalently, using the function notation, we have the indicator function equality $$T(x,y,z) \equiv S_1(x,y) S_2(x,z) S_3(y,z).$$ 
We denote $T=\CI(S_1,S_2,S_3)$.
\end{definition}

It is well-known that if a function $D$ has a non-deterministic NOF protocol which sends $b$ bits, then $D$ can be expressed as the union of $2^b$ cylinder intersections. It thus suffices to focus on the structure of cylinder intersections, and their relation to pseudorandom sets $D$. Our main technical theorem in this context shows that any sparse cylinder intersection must also be sparse with respect to a pseudorandom set $D$.

\begin{theorem}[Pseudorandom sets for cubes are pseudorandom for cylinder intersections]
\label{psd-ci}
Let $D \subset [N]^3$ be $\gamma$-pseudorandom with respect to large cubes. Then, $D$ is pseudorandom with respect to cylinder intersections in the following one-sided sense. Let $t \le \log(1/\gamma)$. For any cylinder intersection $F$ of density
$$
\E_{(x,y,z) \in [N]^3} F(x,y,z) \le 2^{-t},
$$
it holds
$$
\E_{(x,y,z) \in D} F(x,y,z) \leq 2^{-c t^{1/3}}
$$
for some (absolute) constant $c>0$.
\end{theorem}

\Cref{psd-ci} shows that any set $D$ which is sparse and pseudorandom for large cubes, must be difficult for non-deterministic NOF protocols. \Cref{nondet-lb-psd-sparse} follows easily from the above theorem.

\begin{proof}[Proof of \Cref{nondet-lb-psd-sparse} from \Cref{psd-ci}] Assume that a non-deterministic NOF protocol which sends $b$ bits decides $D(x,y,z)$. This implies that $D$ can be expressed as the union of $2^b$ cylinder intersections $F_1,\ldots,F_{2^b} \subset [N]^3$. For each cylinder intersection $F_i$, since $F_i \subset D$ and we assume $|D| \le N^{3-c}$, we have $|F_i| \le N^{3-c}$. Since we assume that $\gamma=N^{-c}$, we may apply \Cref{psd-ci} for $t=c\log N$. This implies that for some constant $c'>0$,
$$
|F_i \cap D| \le 2^{-c' (\log N)^{1/3}} |D|.
$$
Since $F_1,\ldots,F_{2^b}$ cover $D$, their number must satisfy $2^b \ge 2^{c' (\log N)^{1/3}}$, from which the theorem follows.
\end{proof}

We next describe the main ideas in the proof of \Cref{psd-ci}. Its proof relies on a new technical result, that shows that the product of pseudorandom matrices is close to uniform.

\subsection{Uniformity from spreadness}
\label{subsec:uniformity-from-spreadness}

We consider matrices with bounded entries. Given a matrix $A$ with rows indexed by $X$, columns indexed by $Y$ and entries bounded in $[0,1]$, it will be convenient for us to identify it with the function $A:X \times Y \to [0,1]$. We study three properties of matrices each of which measures how close a matrix is to a constant matrix. The first property is \emph{spreadness}, which means that the density of a matrix cannot be significantly increased by restricting to a large rectangle.

\begin{definition}[Spreadness]
Let $A:X \times Y \to [0,1]$, and let $r \ge 1, \eps \in (0,1)$.
We say that $A$ is $(r,\eps)$-spread if for any rectangle $R=X' \times Y' \subset X \times Y$ of size $|R| \ge 2^{-r} |X \times Y|$, it holds that
$$
\E_{(x,y) \in R} [A(x,y)] \le (1+\eps) \E[A].
$$
\end{definition}

The second property is \emph{left lower-bounded}, which in matrix notation means that row averages are not much lower than the global average of the matrix.

\begin{definition}[Left lower-bounded]
Let $A:X \times Y \to [0,1]$, and let $\eps \in (0,1)$.
We say that $A$ is $\eps$-left lower-bounded if
$$
\E_{y \in Y} [A(x,y)] \ge (1-\eps) \E[A] \quad \forall x \in X.
$$
\end{definition}

The third property is \emph{uniformity}, which means that nearly all the entries of a matrix are multiplicatively close to some fixed value.

\begin{definition}[Uniformity]
Let $A:X \times Y \to [0,1]$, and let $k \ge 1, \alpha, \eps \in (0,1)$.
We say that $A$ is $(\alpha, k,\eps)$-uniform if
$$
\Pr_{(x,y) \in X \times Y}[(1-\eps) \alpha \le A(x,y) \le (1+\eps) \alpha] \ge 1-2^{-k}.
$$
\end{definition}

Given $A:X \times Z \to [0,1]$, $B:Y \times Z \to [0,1]$, define their normalized product  $A \circ B:X \times Y \to [0,1]$ as
$$
(A \circ B)(x,y) = \E_{z \in Z} A(x,z) B(y,z).
$$
If we consider $A,B$ as matrices, then $A \circ B = \frac{1}{|Z|} A B^\top$. 
The following theorem, which can be considered as the main new technical step in the proof, shows that if $A,B$ are both spread and left lower-bounded, then their normalized product $A \circ B$ is uniform, where almost all its elements are close to the value expected in the random case, namely $\E[A] \E[B]$.

\begin{theorem}[Product of spread matrices is uniform] \label{spread-to-uniform}
Let $A:X \times Z \to [0,1], B:Y \times Z \to [0,1]$. Let $d,k \ge 1$ and $\eps \in (0,1/80)$. Assume that
\begin{enumerate}
    \item $\E[A], \E[B] \ge 2^{-d}$.
    \item $A,B$ are $(r,\eps)$-spread for $r \ge \Omega(dk/\eps)$.
    \item $A,B$ are $\eps$-left lower-bounded.
\end{enumerate}
Then $A \circ B$ is $(\E[A] \E[B], k, 80 \eps)$ uniform.
\end{theorem}

We note that a corollary of \Cref{spread-to-uniform} is that $\E[A \circ B] \approx \E[A] \E[B]$. Explicitly, $\E[A \circ B] = (1 \pm 80 \eps) \E[A] \E[B] \pm 2^{-k}$.

\subsection{Cylinder intersection closure}
\label{subsec:cylinder-intersection-closure}

We now turn back to prove \Cref{psd-ci}. It will be more convenient to prove an equivalent version of it, which relates to the \emph{cylinder intersection closure} of sets.

\begin{definition}[Cylinder intersection closure]
Let $T \subset [N]^3$. Its cylinder intersection closure, denoted $\CIC(T)$, is the smallest cylinder intersection containing $T$. Explicitly, if we denote the marginals of $T$ by
$$
T_{XY} = \{(x,y): (x,y,z) \in T\}, \quad T_{XZ} = \{(x,z): (x,y,z) \in T\}, \quad T_{YZ} = \{(y,z): (x,y,z) \in T\},
$$
then its cylinder intersection closure is 
$$\CIC(T) = \CI(T_{XY},T_{XZ},T_{YZ}).$$
In other words, $\CIC(T)$ is the set of all points $(x,y,z)$ such that $(x,y,z'),(x,y',z),(x',y,z) \in T$ for some $x',y',z' \in [N]$.
\end{definition}

The following theorem is equivalent to \Cref{psd-ci}, and more convenient for us to prove (see \Cref{psd-ci-cic-equivalent} for a formal proof of equivalence). 

\begin{theorem}
\label{psd-cic}
Let $D \subset [N]^3$ be $\gamma$-pseudorandom with respect to large cubes. Let $T \subset D$ of size $|T| \ge 2^{-d} |D|$, and assume $\gamma \le 2^{-O(d^3)}$. Then
$$
|\CIC(T)| \ge 2^{-O(d^3)} N^3.
$$
\end{theorem}



We now focus on proving \Cref{psd-cic}. We first fix some notations and conventions.
We will consider sets within cubes, $T \subset X \times Y \times Z$. Given such a set, and a function $f:X \times Y \to [0,1]$, we say that the function $f$ is supported on the XY marginal of $T$ if $\text{supp}(f) \subset T_{XY}$. We similarly define when a function is supported on the XZ, YZ marginals of $f$.

We define a notion of a ``well behaved'' set $T$ inside a cube -- a property which depends only on the marginals of $T$. We will show that for such sets, their cylinder intersection closure must be large.

\begin{definition}[Well-behaved sets]
\label{def:well-behaved}
Let $T \subset X \times Y \times Z$, and let $d \ge 1,r \ge 1,\eps \in (0,1)$. We say that $T$ is $(d,r,\eps)$-well behaved if there exist bounded functions $f: X \times Z \to [0,1]$, $g: Y \times Z \to [0,1]$, $h: X \times Y \rightarrow [0,1]$, supported on the $XZ, YZ$, and $XY$-marginals of $T$, respectively, such that the following conditions hold:
\begin{enumerate}
    \item $\E[f],\E[g],\E[h] \geq 2^{-d}$.
    \item $f,g$ are $(r,\eps)$-spread.
    \item $f,g$ are $\eps$-left-lower bounded.
\end{enumerate}
\end{definition}

The following lemma shows that the cylinder intersection closure of well-behaved sets is large. Its proof is a direct application of \Cref{spread-to-uniform}.

\begin{lemma}[Well behaved sets have large cylinder intersection closure]
\label{well-behaved-large-closure}
Let $T \subset X \times Y \times Z$. Assume that $T$ is $(d,r,\eps)$-well behaved for $d \ge 1$, $r=\Omega(d^2), \eps=O(1)$. Then
$$
|\CIC(T)| \ge 2^{-O(d)} |X||Y||Z|.
$$
\end{lemma}

Typically, given a set $T \subset [N]^3$, it will not be well-behaved. However, if there exists a large cube $C$ such that $T \cap C$ (considered as a subset of the cube $C$) is well-behaved, then we may apply \Cref{well-behaved-large-closure} to $T \cap C$ and obtain that its cylinder intersection closure is large, and hence the same holds for $T$. The next lemma guarantees that such a cube exists.

\begin{lemma}[Finding well-behaved sets]
\label{find-cube}
Let $D \subset [N]^3$ be $\gamma$-pseudorandom with respect to large cubes. Let $d \ge 1, r \ge 1, \eps \in (0,1)$, and assume $\gamma \le 2^{-\Omega(dr/\eps)}$. Let $T \subset D$ of size $|T| \ge 2^{-d} |D|$. Then there is a cube $C \subset [N]^3$ of size $|C| \ge 2^{-O(dr/\eps)} N^3$ such that $T \cap C$ (considered as a subset of the cube $C$) is $(d+2,r,\eps)$-well behaved.
\end{lemma}

Combining \Cref{well-behaved-large-closure} and \Cref{find-cube} proves \Cref{psd-cic}, which to recall is equivalent to \Cref{psd-ci}.

\begin{proof}[Proof of \Cref{psd-cic}]
Let $T \subset D$ of size $|T| \ge 2^{-d} |D|$. Apply \Cref{find-cube} with $r=O(d^2), \eps=O(1)$ which we may since we assume $\gamma \le 2^{-\Omega(d^3)}$. This implies the existence of a cube $C \subset [N]^3$ of size $|C| \ge 2^{-O(d^3)} N^3$ such that $T \cap C$ (considered as a subset of the cube $C$) is $(d+2,r,\eps)$-well behaved. \Cref{well-behaved-large-closure} then gives that
$$
|\CIC(T)| \ge |\CIC(T \cap C)| \ge 2^{-O(d)} |C| \ge 2^{-O(d^3)} N^3.
$$
\end{proof}

\subsection{Construction of sparse pseudorandom sets}
\label{subsec:construction}
The last piece of the puzzle is constructing sets $D \subset [N]^3$ which are pseudorandom with respect to large cubes, sparse, and at the same time easy for randomized NOF protocols. Recall that we identify a set $D$ with its indicator function $D:[N]^3 \to \bits$, which will be the function for which we consider NOF protocols.

In order to define such sets $D$, we first define the notion of \emph{expander-colorings}, which are colorings of the edges of the complete bi-partite graph such that each color class is a good expander (it is a variant of two-source extractors). We will then use them to construct the desired set $D$.

\begin{definition}[Expander-coloring]
Let $\col:[N] \times [N] \to [M]$. We say that $\col$ is an $(N,M,\eta)$-expander coloring if for any sets $X,Y \subset [N]$ of size $|X|,|Y| \ge \eta N$ and any color $m \in [M]$,
$$
|\{(x,y) \in X \times Y: \col(x,y)=m\}| = \frac{|X||Y|}{M} (1 \pm \eta).
$$
\end{definition}

We use expander-colorings to define the desired set $D \subset [N]^3$. Given $\col:[N] \times [N] \to [M]$, we define its corresponding set $D(\col) \subset [N]^3$ as the set of triples, such that the color of each of the pairs is the same:
$$
D(\col) = \{(x,y,z) \in [N]^3: \col(x,y) = \col(x,z) = \col(y,z)\}.
$$

We note that for any expander-coloring $\col$, the associated function $D(\col)$ is easy for randomized NOF protocols, since membership in $D$ can be decided by checking the pairwise equalities between $\col(x,y)$, $\col(x,z)$, and $\col(y,z)$. 

\begin{claim}
\label{rand-nof-easy}
For any $\col:[N] \times [N] \to [M]$, there is a randomized NOF protocol (using public randomness) which sends $O(1)$ bits and computes $D(\col)$.
\end{claim}

We show that if $\col$ is a good expander-coloring, then its associated set $D(\col)$ is pseudorandom with respect to large cubes.

\begin{lemma}[Pseudorandom sets from expander-colorings]
\label{pseudorandomness-from-expander-coloring}
Let $\col:[N] \times [N] \to [M]$ be an $(N,M,\eta)$-expander coloring. Let $\gamma \in (0,1)$, and assume $\eta=O(\gamma^3 M^{-5})$. Then $D(\col)$ is $\gamma$-pseudorandom with respect to large cubes.
\end{lemma}

To conclude, we need an explicit construction of an expander-coloring with good parameters. We show that the inner-product function is a good expander-coloring.

\begin{lemma}[Expander-coloring based on inner-product]
\label{inner-product-expander-coloring}
Let $q$ be a prime power and $k \geq 3$.
Consider the inner-product function $\IP:\F_q^k \times \F_q^k \to \F_q$ given by
$$
\IP(x,y) = \ip{x}{y} = \sum_{i=1}^k x_i y_i.
$$
Then $\IP$ is a $(q^k,q, \eta)$-expander coloring for $\eta=q^{-(k-2)/4}$.
\end{lemma}

We can now combine together all the machinery we developed, and prove our main theorem, \Cref{main}.

\begin{proof}[Proof of \Cref{main}]
Let $q$ be a prime power and $k=34$ (any larger constant would also work). \Cref{inner-product-expander-coloring} gives that the inner-product function $\IP$ on $\F_q^k$ is a $(N,M,\eta)$-expander coloring for $N=q^k, M=q, \eta=q^{-8}$. \Cref{pseudorandomness-from-expander-coloring} gives that $D=D(\IP)$ is $\gamma$-pseudorandom with respect to large cubes for $\gamma=O(q^{-1})$. In particular, this implies that $|D|=\Theta(N^3 q^{-2})$, so $D$ is sparse, concretely $|D|=\Theta(N^{3-c})$ for $c=2/3k$. We may thus apply \Cref{nondet-lb-psd-sparse} and obtain that the non-deterministic NOF complexity of $D$ is $\Omega((\log N)^{1/3})$. Finally, \Cref{rand-nof-easy} shows that the randomized NOF complexity of $D$ is $O(1)$.
\end{proof}

\paragraph{Paper organization.} We start in \Cref{sec:prelim} with some preliminary definitions. The proofs of the lemmas from \Cref{subsec:uniformity-from-spreadness}, \Cref{subsec:cylinder-intersection-closure} and \Cref{subsec:construction} are given in \Cref{sec:uniformity-from-spreadness}, \Cref{sec:cylinder-intersection-closure} and \Cref{sec:construction}, respectively. We discuss some open problems in \Cref{sec:open}.

\paragraph{Acknowledgements.} We thank Sreejato Bhattacharya who found a small bug in a previous version of this paper.

\section{Preliminaries}
\label{sec:prelim}

\paragraph{Notations.}
Given positive numbers $x,y$, we shorthand $x = y \pm \eps$ for $y - \eps \le x \le y + \eps$. We similarly define $x=y(1 \pm \eps)$. Given $x \in \R$ we denote $x_+=\max(x,0)$ and $x_-=\max(-x,0)$ so that $x=x_+-x_-$. Given a function $f:\Omega \to \mathbb{R}$ for some domain $\Omega$, its support is $\text{supp}(f) = \{x \in \Omega: f(x) \ne 0\}$.

\paragraph{Norms and normalizations.}
Given a real-valued random variable $X$, its $k$-th norm is $\|X\|_k := \E[|X|^k]^{1/k}$.
Similarly, for a real-valued function $f : \Omega \rightarrow \R$ defined on some ambient finite set $\Omega$, its $k$-th norm is $\|f\|_{k} :=  \left( \E_{x \in \Omega}|f(x)|^k \right)^{1/k}$. Given two functions $f,g : \Omega \rightarrow \R$, their inner product is
$\ip{f}{g} := \E_{x \in \Omega}[f(x)g(x)]$.

We need the following claim from \cite{kelley2023strong}.

\begin{claim}[\cite{kelley2023strong}]
\label{odd-moments}
    Let $\eps \in [0,\frac{1}{4}]$.
    Suppose $X$ is a real-valued random variable with
    \begin{itemize}
    \item $\|X\|_{k} \geq 2 \eps$ for some even $k \in \N$,
    \item $\E X^j \geq 0$ for all $j \in \N$.
    \end{itemize}
    Then $ \|1 + X\|_p \geq 1 + \eps$ for all integers $p \ge k/\eps$.
\end{claim}

\paragraph{Communication complexity.}
We consider NOF 3-party protocols in three communication models: deterministic, randomized (with public ranodmness), and non-deterministic. They are the analog models for the corresponding two-party models. Formally, the deterministic NOF communication complexity of $f$ is the least number of bits needed to compute the function by a deterministic NOF protocol; the randomized NOF communication complexity of $f$ is the least number of bits needed to compute the function by a randomized NOF protocol, with error at most $1/3$, and with access to public randomness; and the non-deterministic NOF communication complexity of $f$ is the least number of bits needed to compute $f$ by a non-deterministic NOF protocol. For more details see any textbook on communication complexity, for example \cite{kushilevitz1996communication} or \cite{rao2020communication}.

\section{Uniformity from spreadness}
\label{sec:uniformity-from-spreadness}

In this section we prove our main technical tool,
\cref{spread-to-uniform}, which says roughly that for two spread matrices $A$ and $B$, the product $A \circ B$ is uniform.
There are two main steps to the argument, and there is a strong analogy between these steps with the arguments involved in the ``sifting" and ``spectral positivity" sections of \cite{kelley2023strong}. However, the proof we give will be self-contained.\footnote{That is, self-contained except for the use of a basic fact about random variables with non-negative odd moments, \cref{odd-moments}.}

\subsection{Grid norms}

We first discuss the following quantity (the ``grid norm" of a matrix)
which is central to the argument.
This quantity is particularly useful for succinctly capturing certain kinds of double-counting arguments involving dense bipartite graphs. For example, the proof of Lemma 7.4 in \cite{gowers2001} involves (implicitly) the quantity $U_{2,5}(A)$, where $A$ is the adjacency matrix encoding the set system involved.
In addition, it is known in various contexts that this quantity can be reasonably interpreted as some kind of measure of pseudorandomness ($A$ is ``pseudorandom" when $\|A\|_{U(\ell,k)}$ is not much larger than $\|A\|_{U(1,1)}$) -- see e.g.\ \cite[Theorem 3.1]{gowers2006} which discusses the approximate equivalence of a bound on $\|A\|_{U(2,2)}$ with some other measures of pseudorandomness.

\begin{definition}[Grid norms] For a function $f:X \times Z \rightarrow \R$, and $\ell, k \in \N$, let
\begin{align*}
U_{\ell,k}(f) &:= \E_{x_1, x_2 ,\ldots, x_\ell \in X}
\left( \E_{z \in Z} f(x_1,z) f(x_2, z) \cdots f(x_\ell,z) \right)^k \\
&= \E_{z_1, z_2,\ldots, z_k \in Z}
\left(  \E_{x \in X} f(x,z_1)  f(x, z_2) \cdots f(x,z_k) \right)^\ell \\
&= \E_{\substack{x \in X^\ell \\ z \in Z^k }} \prod_{i=1}^\ell \prod_{j=1}^k f(x_i, z_j).
\end{align*}
We also write
$$ \|f\|_{U(\ell,k)} := |U_{\ell,k}(f)|^{1/\ell k}.$$
\end{definition}

The name ``grid norm'' is a bit of a misnomer, as $\|\cdot\|_{U(\ell,k)}$ is not a norm in general; still, the name captures the essence in which we use them.
 It is known that $\|\cdot\|_{U(\ell,k)}$ is a semi-norm (that is, it satisfies a triangle inequality) whenever $\ell,k$ are both even \cite[Theorems 2.8, 2.9]{hatami2010graph} -- a fact which we don't use in this work.

 We record some basic properties of $\|\cdot\|_{U(\ell, k)}$.

\begin{claim}[Monotonicity]
Let $\ell,k,\ell',k' \in \N$, where  $\ell \leq \ell'$, $k \leq k'$.
Let $A:X \times Z \to \R_{\ge 0}$. Then,
$$ \|A\|_{U(\ell,k)} \leq \|A\|_{U(\ell',k)} \;\;\;\; \textnormal{ and } \;\;\;\;  \|A\|_{U(\ell,k)} \leq \|A\|_{U(\ell,k')}. $$
\end{claim}

\begin{proof}
By symmetry it suffices to show that
$ \|A\|_{U(\ell,k)} \leq \|A\|_{U(\ell,k')}. $
So, fix $\ell \in \N$.
For uniformly random $x \in [N]^\ell$, consider the resulting random variable
$$ D =D(x) := \left( \E_{z \in Z} A(x_1,z) A(x_2, z) \cdots A(x_\ell,z) \right)^k. $$
Note that $D$ is non-negative since we assume $A$ is non-negative. We have
$$ \E D \leq \left(\E D^{r} \right)^{1/r}$$
for any $r \geq 1$, which for $r=k'/k$ shows that
\begin{equation*}
\|A\|_{U(\ell,k)} = (\E D)^{1/\ell k}  \leq
(\E D^{k'/k} )^{1/\ell k'} =
\|A\|_{U(\ell,k')}. \qedhere
\end{equation*}
\end{proof}

\begin{lemma}[Decoupling inequality for $U_{2,k}$]\label{decoupling}
Let $f : X \times Z \rightarrow \R$ and $g : Y \times Z \rightarrow \R$.
For even $k \in \N$ we have
$$
\E_{x,y} \left( \E_z f(x,z) g(y,z)  \right)^k \leq U_{2,k}(f)^{1/2} \cdot U_{2,k}(g)^{1/2}.
$$
\end{lemma}
\begin{proof} Let $x \in X, y \in Y, z_1,\ldots,z_k \in Z$. Then
\begin{align*}
\E_{x,y} \left( \E_z f(x,z) g(y,z)  \right)^k &=
\E_{x,y} \E_{z_1, \ldots, z_k} \left[\prod_{i=1}^k f(x,z_i) \cdot \prod_{i=1}^k  g(y,z_i)\right] \\
&=  \E_{z_1, \ldots, z_k} \E_{x} \left[ \prod_{i=1}^k f(x,z_i) \right] \cdot \E_{y} \left[  \prod_{i=1}^k  g(y,z_i) \right] \\
&\leq \sqrt{\E_{z_1,\ldots,z_k}\left(  \E_{x}  \prod_{i=1}^k f(x,z_i) \right)^2   }
\sqrt{\E_{z_1,\ldots,z_k}  \left( \E_{y} \prod_{i=1}^k g(y,z_i) \right)^2   } \\
&= \sqrt{U_{2,k}(f)} \sqrt{U_{2,k}(g)}. \qedhere
\end{align*}
\end{proof}

\begin{claim} \label{positive-semidefinite}
    Let $f : X \times Z \rightarrow \R$, and consider
    $$ M(x,x') := \E_{z \in Z} f(x,z)f(x',z). $$
    For any integer $j \in \N$ we have
    $$ \E_{x,x' \in X} M(x,x')^j \geq 0. $$
\end{claim}
\begin{proof}
We note that
$$ \E M(x,x')^j = U_{2,j}(f), $$
and it is clear from the definition of $U_{2,j}(f)$ that the quantity is non-negative for all real-valued functions $f$:
\begin{align*}
 U_{2,j}(f) &= \E_{z_1, z_2,\ldots,z_j \in Z} \left( \E_{x \in X} f(x,z_1)  f(x, z_2) \cdots f(x,z_k) \right)^2 \geq 0. \qedhere
\end{align*}
\end{proof}

\subsection{Spread matrices have controlled grid norms}

Let
$$ \mathcal{R} = \set{\1_{X'}(x)\1_{Y'}(y) }{X' \times Y' \subseteq X \times Y}$$
denote the set of rectangle indicator functions, and let
$$ \text{conv}(\mathcal{R}) = \set{\sum_{i} c_{i} \1_{R_i}}{\1_{R_i} \in \mathcal{R}, c_i \geq 0, \sum_{i} c_i \leq 1}$$
denote its convex hull.
We also consider the slightly richer class of ``soft rectangles",
$$ \mathcal{\widetilde{R}} = \set{f(x)g(y)}{f : X \rightarrow [0,1], g : Y \rightarrow [0,1]}.$$
Next, we note that any soft rectangle can be expressed as a convex combination of rectangles. In particular, this implies that $ \text{conv}(\mathcal{\widetilde{R}}) = \text{conv}(\mathcal{R}) $.
\begin{claim}
$\mathcal{\widetilde{R}} \subseteq \textnormal{conv}(\mathcal{R})$.
\end{claim}
\begin{proof}
Let $f(x)g(y)$ be a soft rectangle. We can write\footnote{
     If desired, one may obtain a representation as a finite combination of rectangles by noting that there are only finitely many different superlevel sets $\set{x}{f(x) \geq s}$ and $\set{y}{g(y) \geq t}$.
     }
\begin{align*}
     f(x)g(y) = \int_{t=0}^1 \int_{s=0}^1 \ind{f(x) \geq s} \ind{g(y) \geq t} \mathop{ds} dt.  &\qedhere
\end{align*}
\end{proof}

The next lemma says that any non-negative function $D: X \times Y \rightarrow \R_{\geq 0}$ that has a non-trivial correlation with an element of $\text{conv}(\mathcal{R})$ also does so with a rectangle of comparable density.
\begin{claim} \label{pruning}
Let $D : X \times Y \rightarrow \R_{\geq 0}$ and $F \in \textnormal{conv}(\mathcal{R})$;
 suppose that $\|D\|_{\infty} \leq \Delta$ and $\|F\|_{1} \geq \delta$.
If
$$ \ip{\frac{F}{\|F\|_{1}}}{D} \geq 1 + \eps, $$
then there is some rectangle $R$ with
$$ \E_{(x,y) \in R} D(x,y) = \ip{\frac{\1_{R}}{\|\1_{R}\|_1}}{D} \geq 1 + \frac{\eps}{2} $$
and
$$ \frac{|R|}{|X||Y|} = \|1_{R}\|_{1} \geq \frac{\eps \delta}{2 \Delta}. $$
\end{claim}

\begin{proof}
Write $F = \sum_{i} c_i \1_{R_i}.$ We begin by pruning rectangles which are too small: define $F' = \sum_{i} c_i' \1_{R_i}$ via
$$c_i' =
\begin{cases}
    c_{i} &\text{ if } \|\1_{R_i}\|_{1} \geq \tau, \\
    0 &\text{ if } \|\1_{R_i}\|_{1} < \tau
\end{cases}
$$
for some threshold value $\tau$. We note that
$$  \frac{\ip{F'}{D}}{\|F\|_{1}} =
\frac{\ip{F}{D}}{\|F\|_{1}} - \frac{\ip{F-F'}{D}}{\|F\|_{1}}  \geq 1 + \eps - \frac{\|F-F'\|_{1}\|D\|_{\infty}}{\|F\|_{1}} \geq 1 + \eps - \frac{\tau \Delta}{\delta}.
$$
We set $\tau = \frac{\eps \delta}{2 \Delta}$, giving
$$\frac{\ip{F'}{D}}{\|F'\|_{1}} \geq \frac{\ip{F'}{D}}{\|F\|_{1}} \geq 1 + \frac{\eps}{2}.$$
In particular, we must have $\ip{F'}{D} > 0$, which guarantees that $F'$ is not identically zero.
We have
$$ \frac{\ip{F'}{D}}{\|F'\|_{1}} =
\frac{\sum_{i} c_i' \ip{\1_{R_i}}{D} } {\sum_{i} c_i' \|\1_{R_i}\|_1}.$$
By averaging, there is some choice of $R = R_i$ with
$$ \ip{ \frac{\1_{R}}{\|\1_{R}\|_1} }{D} \geq 1 + \frac{\eps}{2}$$
and
\begin{align*}
\|\1_{R}\|_{1} \geq \tau.   &\qedhere
\end{align*}

\end{proof}

\begin{lemma} (Sifting a rectangle) \label{lem:sifting}
Let $A : X \times Y \rightarrow [0,1]$; suppose that $\|A\|_{1} \geq \delta$.
Let $\ell,k \in \N$. If
$$ \|A\|_{U(\ell, k)} \geq (1 + \eps) \|A\|_{1}, $$
then there is some rectangle $R$ with
$$ \E_{(x,y) \in R} A(x,y) \geq \left(1 + \frac{\eps}{2} \right) \|A\|_{1}   $$
and
$$ \frac{|R|}{|X||Y|} = \|\1_{R}\|_{1} \geq \tfrac{1}{2} \cdot \eps \cdot \delta^{\ell k + 1}.$$

In particular, if $A$ is $(r,\eps)$-spread for some $r \geq (\ell k + 1) \log(1/\delta) + \log(1/\eps)$, then
$$\|A\|_{U(\ell,k)} \leq (1+2\eps)\|A\|_{1}.$$
\end{lemma}

\begin{proof}

By assumption we have
\begin{align*}
    \|A\|_{U(\ell,k)}^{\ell k} =
    \E_{\substack{x_1, \ldots, x_{\ell} \in X \\ y_{1}, \ldots, y_k \in Y}}
    \left[
    \prod_{i=1}^\ell \prod_{j=1}^{k} A(x_i, y_j)
    \right]
    \geq (1+\eps)^{\ell k} \|A\|_{1}^{\ell k}.
\end{align*}

Our task is to find a reasonably large rectangle $R$ where $A$ is notably denser than average. Before proceeding with the actual argument, let us offer a not-entirely-accurate picture of how this will be done. For illustration, suppose $A$ is an adjacency matrix of a bipartite graph with parts $X$ and $Y$.
We then look for the desired rectangle among those of the following specific form. For any choice of $x_i$'s and $y_j$'s we can consider the rectangle $R = \Gamma(y_1,y_2,\ldots,y_k) \times \Gamma(x_1,x_2,\ldots,x_\ell) \subseteq X \times Y$, where (e.g.) $\Gamma(y_1,y_2,\ldots,y_k) \subseteq X$ denotes the set of common neighbors of the vertices $y_1, \ldots, y_k$ within our bipartite graph. In the actual argument, we will need to consider also some additional, related choices for $R$ which are not so nicely describable. We then use our assumption on $\|A\|_{U(\ell,k)}$ to argue that one of these choices must succeed. We now proceed with the argument (considering now general $A : X \times Y \rightarrow [0,1]$).

Let us fix some arbitrary ordering on tuples $(i,j) \in [\ell] \times [k]$ (say, the lexicographic ordering),
and consider the prefix-products
$$ \phi_{\leq (i,j)}(A) := \prod_{(i',j') \leq (i,j)} A(x_{i'}, y_{j'}).$$
Thus, we have $\E [\phi_{\leq (1,1)}(A)] = \|A\|_{1}$ and $\E[\phi_{\leq (\ell,k)}(A)] = \|A\|_{U(\ell,k)}^{\ell k}$.
Similarly, let us write
$$ \phi_{<(i,j)}(A) := \prod_{(i',j') < (i,j)} A(x_{i'}, y_{j'}),$$
with the convention $\phi_{<(1,1)}(A) := 1$.
Now consider the telescoping product
$$ \prod_{(i,j) \in [\ell] \times [k]} \frac{\E \phi_{\leq (i,j)} (A)}{\E \phi_{<(i,j)}(A)} = \|A\|_{U(\ell,k)}^{\ell k}
$$
This quantity is at least $(1+\eps)^{\ell k} \|A\|_{1}^{\ell k}$, and so we infer that for some choice of $(i^*,j^*)$ we have
$$   \frac{\E \phi_{\leq (i^*,j^*)} (A)}{\E \phi_{<(i^*,j^*)}(A)}
\geq (1 + \eps) \|A\|_{1}.
$$
At this point we would like to think of $\phi_{<(i^*,j^*)}(A)$ primarily as a function of $x_{i^*}$ and $y_{j^*}$. Let us define
$$ F(x_{i^*}, y_{j^*}) =
\E_{\substack{x_1, \ldots, x_{i^*-1}, x_{i^*+1}, \ldots, x_{\ell} \in X \\ y_{1}, \ldots, y_{j^*-1}, y_{j^*+1}, \ldots, y_k \in Y}} \prod_{(i,j) < (i^*,j^*)} A(x_{i}, y_{j})
$$
so that we may reinterpret
$$ \E \phi_{\leq (i^*,j^*)}(A) = \E \phi_{<(i^*,j^*)}(A) \cdot  A(x_{i^*},y_{j^*}) =
\E_{\substack{x_{i^*} \in X\\ y_{j^*} \in Y}} F(x_{i^*}, y_{j^*}) \cdot A(x_{i^*},y_{j^*}) =
\ip{F}{A} $$
and
$$ \E \phi_{<(i^*,j^*)}(A) = \E[F] = \|F\|_{1}. $$
Thus, we have
$$ \ip{\frac{F}{\|F\|_{1}}}{\frac{A}{\|A\|_{1}}} \geq 1 + \eps.$$
Finally, we argue that $F$ is a convex combination of soft rectangles so that we may finish the proof by applying
\cref{pruning} (to $F$ and $D := A / \|A\|_1$). Indeed, as a function of $x_{i^*}, y_{j^*}$, and for any fixing of the other variables $x_{i}, y_{j}$, the quantity
$$ \prod_{(i,j) < (i^*,j^*)} A(x_{i}, y_{j}) =
\left( \prod_{\substack{(i,j) < (i^*,j^*) \\ i \neq i^* \\ j \neq j^*}} A(x_{i}, y_{j})  \right)
\left( \prod_{\substack{(i,j) < (i^*,j^*) \\ i = i^*}} A(x_{i^*}, y_{j})  \right)
\left( \prod_{\substack{(i,j) < (i^*,j^*) \\ j = j^*}} A(x_{i}, y_{j^*})  \right)
$$

is a soft rectangle: each of the factors $A(x_{i}, y_{j})$ may depend on $x_{i^*}$ or $y_{j^*}$ but not both, and the product of $1$-bounded functions is again a $1$-bounded function.

It remains only to discuss what sort of bounds we have on $\|F\|_{1}$ and $\|D\|_{\infty}$.
Since $A$ is $1$-bounded, we have $\|D\|_{\infty} \leq \frac{1}{\|A\|_{1}} \leq \frac{1}{\delta}$.
It follows also from the $1$-boundedness of $A$ that
$$ \|F\|_{1} = \E \phi_{<(i^*,j^*)} \geq \E \phi_{\leq (\ell,j)} = \|A\|_{U(\ell,k)}^{\ell k} \geq \delta^{\ell k}.$$
Thus, \cref{pruning} provides a rectangle $R = X' \times Y'$ of the desired size, $\|\1_{R}\|_1 \geq \eps \delta^{\ell k + 1} / 2$. \qedhere
\end{proof}

\subsection{Products of $U_{2,k}$-regular matrices}

\begin{lemma}[$U_{2,k}$-regularity of $A$ and $B$ implies uniformity of $A \circ B$]\label{lem:regularity-to-uniformity}
Fix an even integer $k \in \N$,  $\eps \in (0,1/20)$, and set $p = \lceil k/\eps \rceil$.
Let $A: X \times Z \rightarrow \R_{\ge 0}$, $B: Y \times Z \rightarrow \R_{\ge 0}$ be two (nonzero) matrices,
and suppose that
\begin{itemize}
\item $\norm{A}_{U(2,p)} \leq (1 + \eps) \|A\|_1 $,
\item $\norm{B}_{U(2,p)} \leq (1 + \eps) \|B\|_1 $,
\item $A,B$ are $\eps$-left lower bounded.
\end{itemize}
Then, the function $D(x,y) = \frac{(A \circ B)(x,y)}{\E[A] \E[B]}$ is close to uniform on $X \times Y$:
$$\|D - 1\|_{k} \leq 20 \eps.$$
\end{lemma}

\begin{proof}
Note that the statement is scale-invariant with respect to $A$ and $B$. Without loss of generality suppose that $\E[A]=\E[B]=1$.
For brevity, let $A_x,B_y :Z \rightarrow \R$ be the functions $A_x(z) = A(x,z), B_y(z) = B(y,z)$.
For any two functions $\alpha,\beta:Z \rightarrow \R$, we write $\ip{\alpha}{\beta} := \E_z[\alpha(z) \beta(z)]$.
With this notation we can express
$$ D(x,y) = \ip{A_x}{B_y}. $$
For what follows it will be convenient to introduce the notation $a(x) := \E_{z} A(x,z)$ and $b(y) := \E_{z}B(y,z)$.

We proceed to argue that $\|\ip{A_x}{B_y} - 1 \|_k \leq O(\eps)$. We have
$$
\ip{A_x}{B_y} - 1 = \ip{A_x - 1}{B_y - 1} + \ip{A_x-1}{1} + \ip{B_y-1}{1}.
$$
We first consider the latter terms. Note that
$$
\|\ip{A_x-1}{1}\|_k = \|a - 1\|_k \leq \|(a - 1)_{-}\|_k + \|(a - 1)_{+}\|_{k}.
$$
We handle the positive and negative deviations from $1$ separately.
First, since we assume $A$ is $\eps$-left lower bounded we have that $a(x) \geq 1 - \eps$ pointwise, which gives
$$\|(a - 1)_{-}\|_k \leq \|(a - 1)_{-}\|_{\infty} \leq \eps. $$
Second, we note that for uniformly random $x$, the resulting random variable $(a(x)-1)_+$ certainly has non-negative odd moments since it is non-negative, and so with \cref{odd-moments} we can obtain a bound on
$ \|(a-1)_+\|_k $ from a bound on $\|1+(a-1)_+\|_{p}.$
Specifically, we have $1 + (a-1)_+ = a + (a-1)_{-}$ and hence
\begin{align*}
\|1+(a-1)_+\|_{p} \leq \|a\|_{p} + \|(a-1)_{-}\|_{p} \leq \|a\|_{p} + \eps
\end{align*}
and
$$ \|a\|_{p} = \|A\|_{U(1,p)} \leq \|A\|_{U(2,p)} \leq 1 + \eps.$$
We conclude that $\|1+(a-1)_+\|_{p} \leq 1 + 2 \eps$, and hence by \Cref{odd-moments}
$$  \|(a-1)_+\|_k \leq 4 \eps. $$

Overall, we obtain the bound $\|a-1\|_k \leq 5 \eps$.
Similarly, also $\|b-1\|_k \le 5 \eps$.
Therefore, by the triangle inequality for $\| \cdot \|_k$, we have
$$\|\ip{A_x}{B_y} - 1\|_k \leq \|\ip{A_x - 1}{B_y - 1}\|_k + 10 \eps.$$

We now apply the decoupling inequality for $U_{2,k}$ (\cref{decoupling}) to study the main term.
$$\|\ip{A_x - 1}{B_y - 1}\|_k^k = \E_{x,y}\left[ \E_z \Big( (A(x,z) - 1) (B(y,z) - 1) \Big)^k\right] \leq \utk{A-1}^{1/2} \utk{B-1}^{1/2},$$
or equivalently,
$$\|\ip{A_x - 1}{B_y - 1}\|_k \leq \|{A-1}\|_{U(2,k)} \|{B-1}\|_{U(2,k)}.$$

Without loss of generality, suppose $\utk{A-1} \geq \utk{B-1}$ so that $\|{A-1}\|_{U(2,k)}^2 \geq \|\ip{A_x - 1}{B_y - 1}\|_k$. Seeking contradiction, we observe that if $\|\ip{A_x}{B_y}-1\|_k > 20 \eps$,
then $\|\ip{A_x - 1}{B_y - 1}\|_k >  10 \eps$, and so
$$\|{A-1}\|_{U(2,k)}^2 >  10 \eps.$$

Let $M(x,x') = \ip{A_x - 1}{A_{x'} - 1}$, and note that
$$\utk{A-1} = \E_{x,x'} \left(\ip{A_x - 1}{A_{x'} - 1}\right)^k = \|M\|_k^k.$$
Consider the random variable $M = M(x,x')$ arising from a uniform random choice of $x, x' \in X$.
As observed in \cref{positive-semidefinite}, $M$ has non-negative odd moments.
Therefore, by \cref{odd-moments}, $\|1 + M\|_p > 1 + 5 \eps$.
Further,
\begin{align*}
    \ip{A_x}{A_{x'}} &= 1 + \ip{A_x - 1}{1} + \ip{1}{A_{x'}-1} + M(x,x') \\
    &= 1 + M(x,x') + (a(x) - 1) + (a(x') - 1) \\
    &\geq 1 + M(x,x') - (a(x) - 1)_{-} - (a(x') - 1)_{-},
\end{align*}
and so
\begin{align*}
\|\ip{A_x}{A_{x'}}\|_p & \geq \|1 + M \|_p - 2\|(a - 1)_{-}\|_{p} \\
    &\geq \|1 + M \|_p - 2\eps \\
    &> 1 + 3\eps,
\end{align*}
since we already noted that $\|(a - 1)_{-}\|_{\infty} \le \eps$.
Thus, we have $\|\ip{A_x}{A_{x'}}\|_p > 1 + 3\eps$. On the other hand, our regularity assumption on $A$ says that
$$\|\ip{A_x}{A_{x'}}\|_p  = \|A\|_{U(2,p)}^2 \leq (1+\eps)^2 < 1 + 3 \eps,$$
giving a contradiction. Thus, we must in fact have $\|\ip{A_x}{B_y} - 1 \|_k \leq 20 \eps$.
\end{proof}

\begin{proof}[Proof of \cref{spread-to-uniform}]
The proof is a straightforward combination of \cref{lem:sifting} with \cref{lem:regularity-to-uniformity}, followed by an application of Hölder's inequality.
Let $p = \lceil k/\eps \rceil$.
We  would like to apply \cref{lem:sifting} to control $U_{2,p}(A)$ and $U_{2,p}(B)$ with our spreadness assumption, which is indeed possible for $r = cdk/\eps$ where $c>0$ is a sufficiently large constant. We conclude that
$$
\|A\|_{U(2,p)} \leq (1 + 2 \eps) \|A\|_{1}, \qquad \|B\|_{U(2,p)} \leq (1 + 2 \eps) \|B\|_{1}.
$$
Consider
$$ D(x,y) := \frac{(A \circ B)(x,y)}{\E[A] \E[B]}.$$
From \cref{lem:regularity-to-uniformity} we have
$$ \|D - 1\|_{k} \leq 40 \eps. $$
For any subset $S \subseteq X \times Y$ of size $|S| \geq 2^{-k} |X \times Y|$ we have
\begin{align*}
 \E_{(x,y) \in S} |D(x,y) - 1| &\leq \left( \E_{(x,y) \in S} |D(x,y) - 1|^{k} \right)^{1/k} \\
 &\leq \left( 2^{k}  \E_{(x,y) \in X \times Y} |D(x,y) - 1|^{k} \right)^{1/k} \\
 &= 2 \|D - 1\|_{k} \\
 &\leq 80 \eps.
\end{align*}
If we consider the set $T \subseteq X \times Y$ of large deviations
$$ T = \set{(x,y)}{|D(x,y) - 1| > 80 \eps} ,$$
it must be the case that $|T| < 2^{-k}|X \times Y|$. Otherwise, we would obtain a large set $T$ with
$$ \E_{(x,y) \in T} |D(x,y) - 1| > 80 \eps, $$
contradicting the calculation above. \qedhere
\end{proof}

\subsection{Correlations involving spread matrices}

Before continuing we record the following (immediate) corollary of \Cref{spread-to-uniform}: if $f$ and $g$ are functions which are dense, spread, and left-lower bounded, and $h$ is any function which is dense, then the quantity $\E_{x,y,z}f(x,z)g(y,z)h(x,y)$ behaves roughly as if the three terms were independent.

\begin{corollary}
\label{three-funcs}
Let $f: X \times Z \to [0,1]$, $g:Y \times Z \to [0,1]$, $h:X \times Y \to [0,1]$. Let $d \ge 1$ and $\delta \in (0,1)$, and set $r=\Omega((d+\log(1/\delta)) d/\delta)$ and $\eps=\delta/160$.
Assume that:
\begin{enumerate}
    \item $\E[f], \E[g], \E[h] \ge 2^{-d}$.
    \item $f,g$ are $(r,\eps)$-spread.
    \item $f,g$ are $\eps$-lower bounded.
\end{enumerate}
Then
$$
\E_{x \in X, y \in Y, z \in Z} \left[ f(x,z) g(y,z) h(x,y) \right] = (1 \pm \delta) \E[f] \E[g] \E[h].
$$
\end{corollary}

\begin{proof}
Define
$$
S = \left\{(x,y) \in X \times Y: (f \circ g)(x,y) = (1 \pm \delta/2) \E[f] \E[g]\right\}.
$$
Applying \Cref{spread-to-uniform} (with $A=f$,$B = g$, $\eps=\delta/160$, and $k=3d + \log(1/\delta)+1$) gives that $|S| \ge (1-2^{-k}) |X| |Y|$. For $(x,y) \in S$ we have
$$
\E_{z \in Z} \left[ f(x,z) g(y,z) h(x,y) \right] = (f \circ g)(x,y) h(x,y) = (1 \pm \delta/2) \E[f] \E[g] h(x,y).
$$
For $(x,y) \notin S$ we can naively bound
$$
\E_{z \in Z} \left[ f(x,z) g(y,z) h(x,y) \right] \in [0,1].
$$
Averaging over all $(x,y) \in X \times Y$ gives
$$
\E_{x \in X, y \in Y, z \in Z} \left[ f(x,z) g(y,z) h(x,y) \right] = (1 \pm \delta/2) \E[f] \E[g] \E[h] \pm \Pr[(x,y) \in S].
$$
The claim follows by the choice of $k$, since
\begin{equation*}
\Pr[(x,y) \in S] \le 2^{-k} \le (\delta/2) 2^{-3d} \le (\delta/2) \E[f] \E[g] \E[h]. \qedhere
\end{equation*}
\end{proof}

\section{Cylinder intersection closure}
\label{sec:cylinder-intersection-closure}

We prove in this section the two lemmas in \Cref{subsec:cylinder-intersection-closure}: \Cref{well-behaved-large-closure,find-cube}, as well as formally show that \Cref{psd-ci} and \Cref{psd-cic} are equivalent.

\begin{claim}
\label{psd-ci-cic-equivalent}
\Cref{psd-ci} and \Cref{psd-cic} are equivalent.
\end{claim}

\begin{proof}
Briefly, \Cref{psd-cic} is the contra-positive form of \Cref{psd-ci}.
Specifically, suppose \Cref{psd-cic} is true. Suppose we have the conditions of \Cref{psd-ci}. Now, set $T = F \cap D$ and $d =
c_1 t^{1/3}$ for a small enough constant $c_1$. If $|T| \geq 2^{-d} |D|$, then we must have $|\CIC(T)| \geq 2^{-c_1 c_2 d^3} N^3$ for some constant $c_2$. This violates the density of $F$ for a suitable constant $c_1$.
The reverse direction follows similarly.
\end{proof}

The proof of \Cref{well-behaved-large-closure} is a straightforward application of \Cref{three-funcs}.

\begin{proof}[Proof of \Cref{well-behaved-large-closure}]
Let $M=\CIC(T)$. Let $f: X \times Z \to [0,1]$, $g: Y \times Z \to [0,1]$, $h: X \times Y \rightarrow [0,1]$ given by the condition that $T$ is $(d,r,\eps)$-well behaved. As $f,g,h$ are bounded and supported on the $XZ, YZ$, and $XY$-marginals of $T$, respectively, we have the pointwise lower-bound
$$
M(x,y,z) \ge f(x,z) g(y,z) h(x,y).
$$
Apply \Cref{three-funcs} to $f,g,h$ to conclude that
$$
\E_{(x,y,z) \in X \times Y \times Z} M(x,y,z) \ge \E_{(x,y,z) \in X \times Y \times Z} f(x,z) g(y,z) h(x,y) \ge 2^{-O(d)}.
$$
This implies that $|M| \ge 2^{-O(d)} |X||Y||Z|$.
\end{proof}

We now move to prove \Cref{find-cube}. We start with the following claim, which shows how the pseudorandomness of $D$ allows to approximate certain averages of ratios that come up in the proof.

\begin{claim}
\label{approx-h}
Assume that $D \subset [N]^3$ is $\gamma$-pseudorandom with respect to a cube $C = X \times Y \times Z$, and let $T \subset D$. Define the function $h:X \times Y \to [0,1]$ given by
$$
h(x,y) = \frac{\E_{z \in Z}[T(x,y,z)]}{\E_{z \in Z}[D(x,y,z)]}.
$$
Then
$$
\E[h] = \frac{|T \cap C|}{|D \cap C|} \pm O(\gamma^{1/3}).
$$
\end{claim}

\begin{proof}
Let $\mu(C) = |D \cap C|/|C|$. Define $v:X \times Y \to \R_{\ge 0}$ by
$$
v(x,y) = \frac{\E_{z \in Z}[D(x,y,z)]}{\mu(C)}.
$$
Note that $\E[v]=1$, and the second moment of $v$ is
$$
\E[v^2] = \frac{\E_{(x,y,z,z') \in X \times Y \times Z \times Z} [D(x,y,z) D(x,y,z')]}{\E_{(x,y,z) \in X \times Y \times Z} [D(x,y,z)]^2}.
$$
The assumption that $D$ is $\gamma$-pseudorandom with respect to $C$ implies that $\E[v^2] = \frac{1 \pm \gamma}{1 \pm \gamma} \le 1+4 \gamma$ and hence $\text{Var}[v]=O(\gamma)$. Define
$$
S = \left\{(x,y) \in X \times Y: v(x,y) = 1 \pm \gamma^{1/3} \right\}.
$$
Chebyshev's inequality gives that $|S| \ge (1-O(\gamma^{1/3})) |X||Y|$.
For $(x,y) \in S$ we have
$$
h(x,y) = \frac{\E_{z \in Z}[T(x,y,z)]}{ (1 \pm \gamma^{1/3}) \mu(C)} = (1 \pm 2 \gamma^{1/3}) \frac{\E_{z \in Z}[T(x,y,z)]}{\mu(C)} = \frac{\E_{z \in Z}[T(x,y,z)]}{\mu(C)} \pm O(\gamma^{1/3}),
$$
where we used the fact that $\E_{z \in Z}[T(x,y,z)] \le \E_{z \in Z}[D(x,y,z)] = (1 \pm \gamma^{1/3}) \mu(C) \le 2 \mu(C)$.
For $(x,y) \notin S$ we naively bound $h(x,y) \in [0,1]$. Thus we can estimate
$$
\E[h] = \frac{\E_{(x,y,z) \in C}[T(x,y,z)]}{\mu(C)} \pm O(\gamma^{1/3}) \pm \Pr[(x,y) \notin S] =  \frac{|T \cap C|}{|D \cap C|} \pm O(\gamma^{1/3}).
$$
\end{proof}

We prove \Cref{find-cube} in two steps. First, we do a density increment to find a cube in which the set $T$ is ``mostly'' well-behaved (with respect to some specific candidate functions $f,g,h$ which are obtained by considering marginals of the uniform distribution on $T$). The only deficiency will be that not all points will be left-lower bounded, but instead only most of them. Then we do a pruning phase to remove the bad points. We start with the necessary definitions.

\begin{definition}[Mostly left lower-bounded]
Let $f:X \times Y \to [0,1]$, and let $\eps \in (0,1), \beta \in (0,1)$.
We say that $f$ is $\beta$-mostly $\eps$-left lower-bounded if for at least a $(1-\beta)$-fraction of $x \in X$, it holds that
$$
\E_{y \in Y} [f(x,y)] \ge (1-\eps) \E[f].
$$
\end{definition}

\begin{definition}[Mostly well-behaved sets]
\label{def:mostly-well-behaved}
Let $T \subset X \times Y \times Z$, and let $d \ge 1,r \ge 1,\eps \in (0,1),\beta \in (0,1)$. We say that $T$ is $(d,r,\eps,\beta)$-mostly well behaved if there exist bounded functions $f: X \times Z \to [0,1]$, $g: Y \times Z \to [0,1]$, $h: X \times Y \rightarrow [0,1]$, supported on the $XZ, YZ$, and $XY$-marginals of $T$, respectively, such that the following conditions hold:
Suppose that
\begin{enumerate}
    \item $\E[f],\E[g],\E[h] \geq 2^{-d}$.
    \item $f,g$ are $(r,\eps)$-spread.
    \item $f,g$ are $\beta$-mostly $\eps$-left-lower bounded.
\end{enumerate}
\end{definition}

\begin{lemma}[Finding mostly well-behaved sets]
\label{find-cube-mostly}
Let $D \subset [N]^3$ be $\gamma$-pseudorandom with respect to large cubes. Let $d \ge 1, r \ge 1, \eps \in (0,1),\beta \in (0,1)$, and assume $\gamma \le 2^{-\Omega(dr/\eps)} \beta$. Let $T \subset D$ of size $|T| \ge 2^{-d} |D|$. Then there is a cube $C \subset [N]^3$ of size $|C| \ge 2^{-O(dr/\eps)} N^3$ such that $T \cap C$ (considered as a subset of the cube $C$) is $(d+1,r,\eps,\beta)$-well behaved.
\end{lemma}

\begin{proof}
Let $\eta=O(\eps/r)$. Given a cube $C \subset [N]^3$ define the function
$$
\phi(C) = \frac{|T \cap C|}{|D \cap C|} \cdot |C|^{\eta}.
$$
Let $C = X \times Y \times Z$ be the cube which maximizes $\phi(\cdot)$. We will show that $C$ satisfies the required properties. Define the functions $f: X \times Z \to [0,1]$, $g: Y \times Z \to [0,1]$, $h: X \times Y \rightarrow [0,1]$ as follows:
$$
f(x,z) = \frac{\E_{y \in Y} T(x,y,z)}{\E_{y \in Y} D(x,y,z)}, \quad
g(y,z) = \frac{\E_{x \in X} T(x,y,z)}{\E_{x \in X} D(x,y,z)}, \quad
h(x,y) = \frac{\E_{z \in Z} T(x,y,z)}{\E_{z \in Z} D(x,y,z)}.
$$
Observe that indeed $f,g,h$ are supported on the $XZ, YZ, XY$ faces of $T$, respectively; and that since $T \subset D$, the functions $f,g,h$ take values in $[0,1]$.
We will prove that $f,g,h$ are all $2^{-d}$ dense, $(r,\eps)$-spread and $\beta$-mostly $\eps$-left lower-bounded. For concreteness, we prove these properties for $h$, but they hold for $f,g$ by an analogous argument.

\paragraph{Large cube.}
First, for $C_0=[N]^3$ we have $\phi(C_0) = (|T|/|D|) N^{3 \eta} \ge 2^{-d} N^{3 \eta}$, and for $C$ we have $\phi(C)=(|T \cap C| / |D \cap C|) |C|^{\eta}$. Since $C$ maximizes $\phi(\cdot)$ we can already make two deductions: first, since $|C| \le N^3$ we must have $|T \cap C| \ge 2^{-d} |D \cap C|$; and second, since $|T \cap C| \le |D \cap C|$ we have $|C|^{\eta} \ge 2^{-d} N^{3 \eta}$, which implies $|C| \ge 2^{-d/\eta} N^3$. Note that our assumption that $\gamma \le 2^{-O(dr/\eps)}$ implies that $D$ is $\gamma$-pseudorandom with respect to $C$.

\paragraph{Density.}
We prove that $\E[h] \ge 2^{-(d+1)}$.
Apply \Cref{approx-h} to $T$ and the cube $C$. We get
$$
\E[h] = \frac{|T \cap C|}{|D \cap C|} \pm O(\gamma^{1/3}) \ge 2^{-d} \pm O(\gamma^{1/3}).
$$
The claim follows since $\gamma \le 2^{-O(d)}$.

\paragraph{Spreadness.}
We next show that $h$ is $(r,\eps)$-spread. Assume towards a contradiction that there exists a rectangle
$R' = X' \times Y' \subset X \times Y$ of size $|R'| \ge 2^{-r} |X| |Y|$ such that
$$
\E_{(x,y) \in X' \times Y'} h(x,y) > (1+\eps) \E_{(x,y) \in X \times Y} h(x,y).
$$
Define $C' = X' \times Y' \times Z$. Since we take $\gamma$ small enough (concretely, $\gamma \le 2^{-(d/\eta+r)}$), $D$ is also $\gamma$-pseudorandom with respect to $C'$. We will show that $\phi(C')>\phi(C)$, which contradicts the assumption that $C$ maximizes $\phi(\cdot)$. Applying \Cref{approx-h} to $C$ and $C'$ gives
$$
\E_{(x,y) \in X \times Y} h(x,y) = \frac{|T \cap C|}{|D \cap C|} \pm O(\gamma^{1/3})
$$
and
$$
\E_{(x,y) \in X' \times Y'} h(x,y) = \frac{|T \cap C'|}{|D \cap C'|} \pm O(\gamma^{1/3})
$$
As we have $\gamma \le (2^{-d} \eps)^{O(1)}$, we get that
$$
\frac{|T \cap C'|}{|D \cap C'|} \ge (1 + \eps/2) \frac{|T \cap C|}{|D \cap C|}
$$
which gives
$$
\frac{\phi(C')}{\phi(C)} \ge (1 + \eps/2) \left(\frac{|C'|}{|C|}\right)^{\eta} \ge (1 + \eps/2) 2^{-r \eta} > 1
$$
for $\eta=O(\eps/r)$ small enough.

\paragraph{Mostly left lower-bounded.}
We next show that $h$ is mostly left lower-bounded. Assume towards a contradiction that there exists $X' \subset X$ of size $|X'| = \beta |X|$ such that
$$
\E_{y \in Y} h(x,y) < (1-\eps) \E[h] \qquad \forall x \in X'.
$$
Set $C' = X' \times Y \times Z$. Since we assume that $\gamma$ is small enough (concretely, $\gamma \le 2^{-d/\eta} \beta $), $D$ is $\gamma$-pseudorandom with respect to $C'$. Applying \Cref{approx-h} to $C'$ gives
$$
\E_{(x,y) \in X' \times Y} h(x,y) = \frac{|T \cap C'|}{|D \cap C'|} \pm O(\gamma^{1/3}).
$$
Repeating the same argument for $C$, and using the fact that we have $\gamma \le (2^{-d} \eps)^{O(1)}$, we get
$$
\frac{|T \cap C'|}{|D \cap C'|} \le (1-\eps) \frac{|T \cap C|}{|D \cap C|} + O(\gamma^{1/3}) \le (1-\eps/2) \frac{|T \cap C|}{|D \cap C|},
$$
Since $D$ is $\gamma$-pseudorandom with respect to $C$, $C'$, we have $|D \cap C| = (1 \pm \gamma) \mu |C|$, $|D \cap C'| = (1 \pm \gamma) \mu |C'|$ and hence
$$
\frac{|T \cap C'|}{|C'|} \le (1-\eps/4) \frac{|T \cap C|}{|C|}.
$$
Let $C'' = C \setminus C' = (X \setminus X') \times Y \times Z$. We will show that $\phi(C'') > \phi(C)$, which is a contradiction to the maximality of $C$. Note that
$$
|T \cap C''| = |T \cap C| - |T \cap C'| \ge (1 - (1-\eps/4) \beta) |T \cap C|
$$
and
$$
|D \cap C''| = |D \cap C| - |D \cap C'| = (1 - \beta \pm \gamma) |D \cap C|
$$
and
$$
|C''| = (1-\beta) |C|.
$$
Thus
$$
\frac{\phi(C'')}{\phi(C)} = \frac{|T \cap C''|}{|T \cap C|} \frac{|D \cap C|}{|D \cap C''|} \left( \frac{|C''|}{|C|}\right)^{\eta} \ge
(1 - \beta + \eps \beta /4) (1 - \beta + \gamma)^{-1} (1-\beta)^{\eta} > 1
$$
where the last inequality holds for $\gamma \le O(\beta \eps), \eta \le O(\eps)$ small enough.
\end{proof}

As we discussed, we prove \Cref{find-cube} by pruning the cube obtained by \Cref{find-cube-mostly}.

\begin{proof}[Proof of \Cref{find-cube}]
Apply \Cref{find-cube-mostly} with parameters $d,r+1,\eps/2,\beta=O(2^{-d} \eps)$, which we can as we assume $\gamma=2^{-\Omega(d^3)}$. Let $C=X \times Y \times Z$ and $f: X \times Z \to [0,1]$, $g: Y \times Z \to [0,1]$, $h: X \times Y \rightarrow [0,1]$ be the obtained functions, satisfying the condition that $T \cap C$ viewed as a subset of $C$ is $(d+1,r+1,\eps/2,\beta)$-well behaved. We next prune $C$ to obtain the desired cube and corresponding functions.

Let $X' \subset X$ be the set of points $x$ where $f$ is $(\eps/2)$-left lower-bounded, and $Y' \subset Y$ be the set of points $y$ where $g$ is $(\eps/2)$-left lower-bounded, both with respect to $C$. Let $C'=X' \times Y' \times Z$ and let $f',g',h'$ be the restrictions of $f,g,h$ to $X' \times Z, Y' \times Z, X' \times Y'$, respectively. We claim that $T \cap C'$, viewed as a subset of the cube $C'$, is $(d+2,r,\eps)$ well-behaved, witnessed by $f',g',h'$.

We first show that $\E[f'],\E[g'],\E[h'] \ge 2^{-(d+2)}$. We show this for $f'$, and an analogous argument works for $g', h'$.
Note that since $|X'| \ge (1-\beta) |X|$ and $f$ takes values in $[0,1]$, we have
$$
\E[f'] = \E[f] \pm O(\beta).
$$
Since we know $\E[f] \ge 2^{-(d+1)}$, taking $\beta=O(2^{-d})$ small enough guarantees that $\E[f'] \ge 2^{-(d+2)}$. 

We next show that $f',g'$ are $(r,\eps)$-spread. We show this for $f'$, and an analogous argument works for $g'$.
Assume that $R \subset X' \times Z$ is a rectangle of size $|R| \ge 2^{-r} |X'| |Z|$. We can also view $R$ as a rectangle $R \subset X \times Z$ of size $|R| \ge (1-\beta) 2^{-r} |X| |Z| \ge 2^{-(r+1)} |X| |Z|$. Recalling that $f'$ is a restriction of $f$, and applying the assumption that $f$ is $(r+1,\eps/2)$-spread gives
$$
\E_{(x,z) \in X' \times Z} f'(x,z) = \E_{(x,z) \in X' \times Z} f(x,z) \le (1+\eps/2) \E[f] \le (1+\eps) \E[f'],
$$
where in the last inequality we use the fact that $\E[f']=\E[f] \pm O(\beta)$ and our choice of $\beta = O(2^{-d} \eps)$. 

Finally, we show that $f',g'$ are $\eps$-left lower-bounded. We show this for $f'$, and an analogous argument works for $g'$.
Take any $x \in X'$. We have by assumption
$$
\E_{z \in Z}[f(x,z)] \ge (1-\eps/2) \E[f].
$$
We already saw that $\E[f'] = \E[f] \pm \beta$, and so for $\beta = O(2^{-d}\eps)$ we get that
$$
\E_{z \in Z}[f'(x,z)] = \E_{z \in Z}[f(x,z)] \ge (1-\eps/2) \E[f] \ge (1-\eps) \E[f'].
$$
\end{proof}

\section{Construction of sparse pseudorandom sets}
\label{sec:construction}

We prove in this section the three lemmas from \Cref{subsec:construction}: \Cref{rand-nof-easy,inner-product-expander-coloring,pseudorandomness-from-expander-coloring}.

\begin{proof}[Proof of \Cref{rand-nof-easy}]
Let $(x,y,z) \in [N]^3$ denote the inputs to $D=D(\col)$. Each player sees two out of three inputs; namely, the first player sees $(y,z)$, the second $(x,z)$, and the third $(x,y)$. Each player computes the color of its respective edge, namely $c_1=\col(y,z)$, $c_2=\col(x,z)$, $c_3=\col(x,y)$. They then need to decide if $c_1=c_2=c_3$. This can be easily done by using a randomized protocol for equality between each pair of players, which requires sending only $O(1)$ bits using public randomness.
\end{proof}

We next show that the inner product function is a good expander-coloring. The proof uses standard arguments based on Fourier analysis.

\begin{proof}[Proof of \Cref{inner-product-expander-coloring}]
Let $\F_q$ be a finite field, $k \ge 3$ and set $\eta=q^{-(k-2)/4}$.
Let $X, Y \subset \F_q^k$ of size $|X|,|Y| \ge \eta q^k$. Fix any value $v \in \F_q$. We need to show that
$$
\Pr_{(x,y) \in X \times Y}[\ip{x}{y}=v] = q^{-1} (1 \pm \eta).
$$
We prove this using Fourier analysis. 
The additive characters of $\F_q$ are $\chi_u:\F_q \to \mathbb{C}^*$ for $u \in \F_q$. Given $a \in \F_q$, we have $\E_u[\chi_u(a)]=\1[a=0]$, and hence
$$
\Pr_{(x,y) \in X \times Y}[\ip{x}{y}=v] = \E_{(x,y) \in X \times Y} \E_{u \in \F_q} \left[\chi_u(\ip{x}{y}-v) \right].
$$
Using the fact that $\chi_0 \equiv 1$, $\chi_u(a+b)=\chi_u(a)\chi_u(b)$ and $|\chi_u(-v)|=1$ we get
$$
\Pr_{(x,y) \in X \times Y}[\ip{x}{y}=v] - q^{-1} \le q^{-1} \sum_{u \in \F_q \setminus \{0\}} \left| \E_{(x,y) \in X \times Y} \left[\chi_u(\ip{x}{y}) \right] \right|.
$$
To conclude the proof we bound the latter sum using Lindsey's lemma (see e.g. \cite{chor1988unbiased}). Fix $u \in \F_q \setminus \{0\}$. Let $1_X, 1_Y \in \bits^{q^k}$ be the indicator vectors of $X,Y$, respectively. Let $H$ be the corresponding Fourier transform matrix over $\F_q^k$, namely $H_{x,y}=\chi_u(\ip{x}{y})$ for $x,y \in \F_q^k$. It is well known that $H H^* = q^k I$ and hence its spectral norm is $\|H\|=q^{k/2}$. Thus
$$
\left| \E_{(x,y) \in X \times Y} \left[\chi_u(\ip{x}{y}) \right] \right|
= \frac{|1_X H 1_Y|}{|X||Y|} \le \frac{\|1_X\|_2 \|1_Y\|_2 \|H\|}{|X||Y|} = \frac{\|H\|}{\sqrt{|X||Y|}} \le \frac{1}{\eta q^{k/2}} = q^{-1} \eta,
$$
where the last equality follows by our choice of $\eta$.
\end{proof}

We now move to prove \Cref{pseudorandomness-from-expander-coloring}. We first develop counting lemmas for expanders, which we then apply to prove the lemma.

\subsection{Counting lemma for bi-partite expanders}

As a starting point, we develop counting lemmas for a single color class. These effectively are bi-partite expanders, but in a slightly non-standard regime, so we formally define them.

Let $G=(U,V,E)$ be a bi-partite graph with parts of equal size $|U|=|V|=N$. Given sets $X \subset U, Y \subset V$, we denote by $e_G(X,Y)$ the number of edges between $X,Y$.

\begin{definition}[Bi-partite expander]
Let $G=(U,V,E)$ be a bi-partite graph with $|U|=|V|=N$.
We say that $G$ is a $(N,p,\eta)$-expander, if for any sets $X \subset U, Y \subset V$ of size $|X|,|Y| \ge \eta N$ it holds
$$
e_G(X,Y) = p|X||Y| (1 \pm \eta).
$$
\end{definition}

Next, we extend the definition to $k$-partite graphs. Let $k \ge 2$, and let $G=(V_1,\ldots,V_k;E)$ be a $k$-partite graph with parts $V_1,\ldots,V_k$, each of size $N$. For $i \ne j$ we denote by $G_{ij}$ the induced bi-partite graph between parts $V_i,V_j$, and shorthand $E_{ij}(G)=E(G_{ij})$.

\begin{definition}[$k$-partite expander]
Let $G$ be a $k$-partite graph, with parts each of size $N$. We say that $G$ is a $k$-partite $(N,p,\eta)$-expander if for any $i \ne j$, the bi-partite graph $G_{ij}$ is an $(N,p,\eta)$-expander.
\end{definition}

A graph $H$ on $k$ nodes is said to be a \emph{labeled graph} if its nodes are labeled by $1,\ldots,k$.

\begin{definition}[Labeled graph homomorphism]
Let $G$ be a $k$-partite $(N,p,\eta)$-expander with parts $V_1,\ldots,V_k$. Let $H$ be labeled graph on $k$ nodes. A tuple $(v_1,\ldots,v_k) \in V_1 \times \cdots \times V_k$ is a homomorphism from $H$ to $G$ if edges of $H$ map to edges of $G$; that is, if it satisfies
 $$
\forall (i,j): \; (i,j) \in E(H) \Rightarrow (v_i, v_j) \in E_{ij}(G).
$$
We denote by $\Hom(H,G)$ the set of all such tuples.
\end{definition}

Our main goal in this section is to prove the following counting lemma.

\begin{lemma}
\label{expander-counting-lemma}
Let $G$ be a $k$-partite $(N,p,\eta)$-expander, let $H$ a labeled graph with $k$ nodes and $\ell$ edges, and let $U_i \subset V_i$ for $i \in [k]$. Then the number of homomorphisms $(v_1,\ldots,v_k)$ from $H$ to $G$, that satisfy $v_i \in U_i$ for all $i$, satisfies
$$
|\Hom(H,G) \cap (U_1 \times \cdots \times U_k)| = p^{\ell} \prod_{i \in [k]} |U_i| \pm 6 \eta \ell N^k.
$$
\end{lemma}

Before proving \Cref{expander-counting-lemma}, we need the following claim.

\begin{claim}
\label{expander-most-nodes-have-correct-degree}
Let $G$ be a $(N,p,\eta)$-expander with parts $U, V$.
Let $X \subset U, Y \subset V$ of size $|X|,|Y| \ge \eta N$. Define
$$
X' = \{x \in X: e_G(\{x\}, Y) = p |Y| (1 \pm \eta)\}.
$$
Then $|X \setminus X'| \le 2 \eta N$.
\end{claim}

\begin{proof}
Define
$$
X_1 = \left\{x \in X: e_G(\{x\},Y) > p|Y| (1 + \eta)\right\}, \qquad  X_2 = \left\{x \in X: e_G(\{x\},Y) < p|Y| (1 - \eta)\right\}.
$$
Since $|Y| \ge \eta N$ we have $|X_1|, |X_2| < \eta N$. The claim follows since $X \setminus X' \subset X_1 \cup X_2$.
\end{proof}

\begin{proof}[Proof of \Cref{expander-counting-lemma}]
We shorthand $\cH = \Hom(H,G) \cap (U_1 \times \cdots \times U_k)$.
The proof is by induction on $k$. We start with some base cases. If $H$ has no edges then clearly $|\cH|=\prod_{i \in [k]} |U_i|$ and the lemma holds. Similarly, if $H$ has an isolated node then the claim easily reduces to $k-1$ by removing this node from $H$, and the corresponding part from $G$, so may assume $H$ has no isolated nodes. Finally, note that if some set $U_i$ has size $|U_i| \le \eta N$, then
$|\cH| \le \prod |U_i| \le \eta N^k$, and the lemma also holds. Thus, we may assume that $|U_i| \ge \eta N$ for all $i \in [k]$.

The base case of the induction is $k=2$, where $H$ consists of a single edge $(1,2)$. In this case, by definition of a $(N,p,\eta)$-expander we have
$$
|\cH| = |\Hom(H,G) \cap (U_1 \times U_2)| = e_G(U_1,U_2) = p |U_1| |U_2| (1 \pm \eta) = p |U_1| |U_2| \pm \eta N^2
$$
and the lemma holds.

We next consider $k \ge 3$. For $v_k \in U_k$, define $\cH(v_k)=\{(v_1,\ldots,v_{k-1}): (v_1,\ldots,v_k) \in \cH\}$. 
Then
$$
|\cH| = \sum_{v_k \in U_k} |\cH(v_k)|.
$$
Assume the node $k \in V(H)$ has $s$ neighbours in $H$, which we may assume without loss of generality are $1,\ldots,s$. Let $\Gamma_i(v_k)=\{v_i \in U_i: (v_i,v_k) \in E_{ik}(G)\}$ denote the neighbours of $v_k$ in $U_i$ for $i \in [s]$. Let $G'$ be the $(k-1)$-partite graph obtained from $G$ by removing the part $V_k$. Let $H'$ be the graph obtained from $H$ by removing node $k$. Observe that $G'$ is a $(k-1)$-partite $(N,p,\eta)$-expander, that $H'$ is a labeled graph with $k-1$ nodes and $\ell-s$ edges, and that
$$
\cH(v_k) = \Hom(H',G') \cap (\Gamma_1(v_k) \times \cdots \times \Gamma_{s}(v_k) \times U_{s+1} \times \cdots \times U_{k-1} ).
$$
Using the induction hypothesis for $G',H'$ gives
$$
|\cH(v_k)| = p^{\ell-s} \prod_{i=1}^{s} |\Gamma_i(v_k)| \cdot \prod_{i=s+1}^{k-1} |U_i| \pm  6 \eta (\ell-s) N^{k-1}.
$$
Next, define
$$
U'_k = \{v_k \in U_k: \forall i \in [s], \; |\Gamma_i(v_k)| = p|U_i| (1 \pm \eta)\}.
$$
Applying \Cref{expander-most-nodes-have-correct-degree}, we have $|U_k \setminus U'_k| \le 2s \eta N$. We now complete the calculations. Note that we may assume $\eta \le 1/6\ell$ otherwise the bound is trivial; in this regime we have $(1 \pm \eta)^{s} = 1 \pm 2 \eta s$. For $v_k \in U'_k$ we have
$$
|\cH(v_k)| = p^{\ell} \prod_{i=1}^{k-1} |U_i| (1 \pm \eta)^s \pm 6(\ell-s) \eta N^{k-1} = p^{\ell} \prod_{i=1}^{k-1} |U_i| \pm (6\ell-3s) \eta N^{k-1}.
$$
For $v_k \in U_k \setminus U'_k$ we naively bound
$$
\sum_{v_k \in U_k \setminus U'_k} |\cH(v_k)| \le |U_k \setminus U'_k| N^{k-1} \le 2 s \eta N^k.
$$
Summing over all $v_k \in U_k$, we conclude that
$$
|\cH| = p^{\ell} \prod_{i=1}^k |U_i| \pm 6 \eta \ell N^{k-1}.
$$
\end{proof}

\subsection{Counting lemma for partite expander-colorings}

We now apply the counting lemma for expanders (\Cref{expander-counting-lemma}) to count monochromatic patterns in expander-colorings.
Let $K_N^{(k)}$ denote the complete $k$-partite graph, with parts $V_1,\ldots,V_k$, each of size $N$. We identify each $V_i$ with $[N]$ when possible to do so without confusion. We consider edge-colorings of $K_N^{(k)}$. We denote them by $\cols=(\col_{ij}: 1 \le i < j \le k)$ where each $\col_{ij}:[N]^2 \to [M]$.

\begin{definition}[$k$-partite expander-coloring]
An edge-coloring $\cols$ of $K_N^{(k)}$ is a $k$-partite $(N,M,\eta)$-expander coloring if $\col_{ij}$ are $(N,M,\eta)$-expander colorings for all $i \ne j.$
\end{definition}

\begin{definition}[Monochromatic patterns]
Let $\cols$ be a $k$-partite $(N,M,\eta)$-expander coloring. Let $H$ be labeled graph on $k$ nodes. A tuple $(v_1,\ldots,v_k) \in [N]^k$ is a monochromatic copy of $H$ in $G$ if all edges of $H$ are mapped to edges with the same color under $\cols$. Namely, if there exists $m \in [M]$ such that
 $$
\forall (i,j): (i,j) \in E(H) \Rightarrow \col_{ij}(v_i, v_j)=m.
$$
We denote by $\Mon(H,\cols)$ the set of all such tuples.
\end{definition}

The following lemma is an application of \Cref{expander-counting-lemma} to count monochromatic patterns inside partite expander-colorings.

\begin{lemma}
\label{colored-expander-counting-lemma}
Let $\cols$ be a $k$-partite $(N,M,\eta)$-expander coloring, let $H$ a labeled graph with $k$ nodes and $\ell$ edges, and let $U_i \subset V_i$ for $i \in [k]$. Then
$$
|\Mon(H,\cols) \cap (U_1 \times \cdots \times U_k)| = M^{1-\ell} \prod_{i \in [k]} |U_i| \pm 3 \eta \ell M N^k.
$$
\end{lemma}

\begin{proof}
For each color $m \in [M]$, let $G_m$ be the $k$-partite graph corresponding to edges in $\cols$ colored in color $m$. By definition, $G_m$ is a $k$-partite $(N,M^{-1},\eta)$-expander. The lemma follows by applying \Cref{expander-counting-lemma} to each $G_m$ and summing over all $m$.
\end{proof}

\subsection{Monochromatic triangles in expander-colorings}

We prove \Cref{pseudorandomness-from-expander-coloring} based on \Cref{colored-expander-counting-lemma}. We need one more definition before giving the proof. Given an expander-coloring $\col:[N] \times [N] \to [M]$, we denote by $\cols(\col, k)$ the $k$-partite expander-coloring where $\col$ is used as the edge-coloring between each of two parts; namely $\cols(\col,k)=(\col_{ij}=\col: 1 \le i<j\le k)$.

\begin{proof}[Proof of \Cref{pseudorandomness-from-expander-coloring}]
Let $\col:[N] \times [N] \to [M]$ be an $(N,M,\eta)$-expander coloring. Let $D=D(\col)$. We will show that $D$ is $\gamma$-pseudorandom if we choose $\eta$ small enough. Let $C=X \times Y \times Z$ be a cube of size $|C| \ge \gamma N^3$. Define $v:X \times Y \to \R_{\ge 0}$ as
$$
v(x,y) = \E_{z \in Z} D(x,y,z).
$$
We need to show that $\E[v] = (1 \pm \gamma) M^{-2}$ and $\E[v^2] = (1 \pm \gamma) M^{-4}$.

We first analyze the first moment of $v$. Define a 3-partite $(N,M,\eta)$-expander coloring $\cols_1=\cols(\col,3)$. Let $H_1$ be a triangle. 
Observe that
$$
\E[v] = \frac{|\Mon(H_1, \cols_1) \cap (X \times Y \times Z)|}{|X||Y||Z|}.
$$
\Cref{colored-expander-counting-lemma} then gives
$$
\E[v] = M^{-2} \pm O(\eta M N^3 / |X||Y||Z|) = M^{-2} \pm O(\eta M / \gamma) = (1 \pm \gamma) M^{-2}
$$
since we assume $\eta=O(\gamma^2 M^{-3})$.

Next, we analyze the second moment of $v$. Define a 4-partite $(N,M,\eta)$-expander coloring $\cols_2=\cols(\col,4)$. Let $H_2$ be a union of two triangles sharing an edge; concretely, $V(H_2)=\{1,2,3,4\}$ and $E(H_2)=\{(1,2),(1,3),(2,3),(1,4),(2,4)\}$.
Observe that
$$
\E[v^2] = \frac{|\Mon(H_2, \cols_2) \cap (X \times Y \times Z \times Z)|}{|X||Y||Z|^2}.
$$
\Cref{colored-expander-counting-lemma} then gives
$$
\E[v^2] = M^{-4} \pm O(\eta M N^4 / |X||Y||Z|^2) = M^{-4} \pm O(\eta M / \gamma^2) = (1 \pm \gamma) M^{-4}
$$
since we assume $\eta=O(\gamma^3 M^{-5})$.
\end{proof}

\section{Open problems}
\label{sec:open}

We view \Cref{spread-to-uniform} as the main new technical innovation of this work, and \Cref{main} as an application of it. 
One natural open problem is to extend the proof to other functions, and in particular to $\exactly$; a strong lower bound for the NOF deterministic complexity of it would imply strong lower bounds for the corners problem. Another open problem is to extend the proof to more than 3 players. The challenges in it appear to be similar to those of extending the results of \cite{kelley2023strong} from three-term arithmetic progressions to longer progressions.

\bibliographystyle{alpha}
\bibliography{nof}

\end{document}